\documentclass [12pt]{article}

\usepackage{amsmath,amsthm,amscd,amssymb}

\usepackage[small, centerlast, it]{caption}
\usepackage{epsfig}
\usepackage[titles]{tocloft}
\usepackage[latin1]{inputenc}
\usepackage{verbatim} 
\usepackage[active]{srcltx} 
\usepackage{fancyhdr}
\usepackage{caption}

\def\b1{{1\!\!1}}

\def\cH{{\ca H}}

\def\cL{{\ca L}}

\def\cP{{\ca P}}

\def\bC{{\mathbb C}}           

\def\bR{{\mathbb R}}

\def\bS{{\mathbb S}}

\def\gP{{\mathfrak P}}

\def\beq{\begin{eqnarray}}
\def\eeq{\end{eqnarray}}

\newcommand{\ca}[1]{{\cal #1}}         

\usepackage{sectsty}
\sectionfont{\normalsize}
\subsectionfont{\normalsize\normalfont\itshape}
\newtheoremstyle{thm}
{12pt}
{12pt}
{\itshape}
{}
{\itshape\bfseries}
{}
{1em}
{}

\theoremstyle{thm}
\newtheorem{theorem}{Theorem}
\newtheorem{lemma}[theorem]{Lemma}
\newtheorem{proposition}[theorem]{Proposition}
\newtheorem{definition}[theorem]{Definition}

\begin{document}

\hfill{\sl  UTM 752  June 2012} 
\par 
\bigskip 
\par 
\rm


\par
\bigskip
\large
\noindent
{\bf  Generalized Complex Spherical Harmonics, Frame Functions, and Gleason Theorem}
\bigskip
\par
\rm
\normalsize


\noindent  {\bf Valter Moretti$^{a}$} and {\bf Davide Pastorello$^{b}$}\\
\par
\noindent Department of  Mathematics, University of Trento, via Sommarive 14, 38123 Povo (Trento), Italy.
\smallskip

\noindent $^a$ moretti@science.unitn.it\qquad $^b$ pastorello@science.unitn.it\\
 \normalsize

\par

\rm\normalsize


\rm\normalsize


\par
\bigskip

\noindent
\small
{\bf Abstract}. Consider  a finite dimensional complex Hilbert space $\cH$, with $dim(\cH) \geq 3$,
 define $\bS(\cH):= \{x\in \cH \:|\: ||x||=1\}$, and let $\nu_\cH$ be
the unique regular   Borel positive measure invariant under the action of the unitary operators in $\cH$, with $\nu_\cH(\bS(\cH))=1$.
We prove that if a complex  frame function $f : \bS(\cH)\to \bC$  
 satisfies  $f \in \cL^2(\bS(\cH), \nu_\cH)$, then it 
verifies Gleason's statement: There is a unique linear operator $A: \cH \to \cH$  such that  $f(u) = \langle u| A u\rangle$ for every $u \in \bS(\cH)$. $A$ is Hermitean when 
$f$ is real. No boundedness requirement is thus assumed on $f$  {\em a priori}.
\normalsize

\section{Introduction}
In the absence of superselection rules, the states of a quantum system described in the Hilbert space $\cH$ are defined as generalized probability 
measures $\mu:\gP(\cH)\rightarrow [0,1]$ on the lattice $\gP(\cH)$ of orthogonal projectors in $\cH$. By definition $\mu$ is required to verify (i) $\mu(I)=1$ 
and  (ii) $\mu(\sum_{k\in K} P_k)=\sum_{k\in K} \mu (P_k)$,
where $\{P_k\}_{k\in K} \subset \mathfrak P(\mathcal H)$, with $K$ finite or countable,  is any set satisfying $P_i P_j=0$ for $i\not =j$
and the sum in the right-hand side in (ii) is computed respect to the strong operator topology if $K$ is infinite. \\
Normalized, positive trace-class operators, namely {\em density} or {\em statistical} operators, very familiar to physicists,  define such measures. However, the 
complete characterization of those measures was established by
 Gleason \cite{Gleason}, with a milestone theorem whose proof is unexpectedly difficult.

\begin{theorem}\emph{{\bfseries (Gleason's theorem)}} \label{Gleason}
Let $\cH$ be a (real or complex)  separable Hilbert space with $3\leq dim(\cH) \leq +\infty$. For every generalized probability measures $\mu:\mathfrak P(\cH)\rightarrow [0,1]$, 
there exist a unique positive, self-adjoint trace class operator $T_\mu$, with unit trace, such that:
$$\mu(P)=tr(T_\mu P) \quad \forall P\in \mathfrak P(\cH).$$
\end{theorem}

\noindent The key-tool exploited in Gleason's proof is the notion of {\em frame function} that will be
the object of this paper. 

\begin{definition}\label{ff}
Let  $\cH$ be a complex Hilbert space and  
$\bS(\cH):= \{\psi\in \cH\:| \: ||\psi|| =1\}$. \\
$f:\bS(\mathcal H)\rightarrow \bC$ is a {\bf frame function} on $\cH$ if  $W_f\in \bC$ exists, 
called {\bfseries weight} of $f$, with:
\begin{equation}\label{weight}
\sum_{x\in N} f(x)=W_f \quad \mbox{for every Hilbertian basis $N$ of $\cH$.}
\end{equation}
(If $\cH$ is non-separable, the series is the integral 
with respect to the measure counting the points of $N$.)
\end{definition} 

\noindent With the hypotheses of Gleason's theorem, the restriction $f_\mu$ of $\mu$ to the set of the projectors 
on one-dimensional 
subspaces is a {\em real} and {\em bounded} frame function.
It is known that on a real  Hilbert space $\cH$ with $dim(\cH)= 3$, a frame function which is bounded (even from below only  or only from above only )   is continuous and can be uniquely represented 
as a quadratic form \cite{Gleason,Dvu}.  
That result is very difficult to be established and is the kernel of the original proof of the Gleason theorem. 
The last non-trivial step in Gleason's proof is passing from $3$ real dimensions to any (generally complex) dimension,  this is done exploiting Riesz theorem,
establishing that there is a unique positive, self-adjoint trace-class operator $T_\mu$ with $tr(T_\mu)=1$ such that $f_\mu(x) = \langle x| T_\mu x\rangle$ for all $x\in \bS(\cH)$.
The final step is the easiest one: if $P\in \gP(\mathcal H)$, there is a Hilbert basis $N$ such that
in the strong operator topology $P= \sum_{z \in N_P} z \langle z|\cdot \rangle$ for some $N_P \subset N$.
So that $\mu(P) = \sum_{z\in N_P} f(z) = tr(PT_\mu)$.

Frame functions are therefore remarkable tools to manipulate generalized measures. However, they are interesting on their own right \cite{Dvu}. An important difference, distinguishing the finite-dimensional case from the infinite-dimensional one, 
is that a frame function on an {\em infinite} dimensional Hilbert space has to be automatically 
bounded \cite{Dvu}. Whereas in the finite-dimensional case  ($dim(\cH) \geq 3$), as proved by Gudder and Sherstnev,
there exist infinitely many unbounded frame functions \cite{Dvu}. The bounded ones are the only  representable 
as quadratic forms.

In the rest of the paper we prove a proposition concerning sufficient conditions to assure that a frame function, on a 
complex finite-dimensional Hilbert space $\cH$, with $dim(\cH) \geq 3$, is 
representable as a quadratic form without assuming the boundedness requirement {\em a priori}. Instead 
we treat the topic from another point of view. The sphere $\bS(\cH)$, up to a multiplicative constant, admits a unique regular Borel measure
invariant under the action of all unitary operators in $\cH$. We prove that, for $dim(\cH) \geq 3$, a
complex frame function $f$  is representable as a quadratic form whenever it  is Borel-measurable and 
belongs to $\cL^2(\bS(\cH), \nu_\cH)$. In particular it holds when $f\in \cL^p(\bS(\cH), \nu_\cH)$  for some $p\in [2,+\infty]$.
The proof is direct  and relies upon the properties of the spaces of generalized complex spherical harmonics
\cite{rudin} and  on some results due to Watanabe \cite{armgenfg} on zonal harmonics, beyond standard facts on Hausdorff compact 
topological group representations (the classic Peter-Weyl theorem).

\section{Generalized Complex Spherical Harmonics}
Let us introduce a $n$-dimensional generalization of  spherical harmonics defined on:
\begin{equation}
\bS^{2n-1} :=\{x\in\bC^n \:|\: ||x||=1\}\:.\label{prod}
\end{equation}
 $\bS^{2n-1} \subset \bR^{2n}$ is, in fact,  a $2n-1$-dimensional real smooth manifold.\\
$\nu_n$ denotes  the $U(n)$-left-invariant regular 
Borel measure on $\bS^{2n-1}$, normalized to $\nu_n(\bS^{2n-1})=1$, obtained from the two-sided 
Haar measure
on $U(n)$ on the  homogeneous space given by the quotient $U(n)/U(n-1) \equiv \bS^{2n-1}$. That measure 
exists and is unique as follows from general results by Mackey (e.g., see  Chapter 4 of \cite{BR},
noticing that both $U(n)$ and $U(n-1)$ are compact and thus unimodular). 

\begin{lemma}\label{lemma0}
$\nu_n (A)>0$ if $A \neq \emptyset$ is an open subset of $\bS^{2n-1}$.
\end{lemma}

\begin{proof}
$\{gA\}_{g\in U(n)}$ is an open covering of  $\bS^{2n-1}$. Compactness implies that $\bS^{2n-1} = \cup_{k=1}^ N g_kA$
for some finite $N$. If   $\nu_n (A)=0$, sub-additivity and  $U(n)$-left-invariance would imply $\nu_n(\bS^{2n-1})=0$ that is false.  
\end{proof}
\noindent As $\nu_n$ is  $U(n)$-left-invariant,
\begin{equation}\label{rep}
U(n)\ni g\to D_n(g)\quad \mbox{with $D_n(g)f:=f\circ g^{-1}$ for $f\in L^2(\bS^{2n-1},d\nu_n)$} 
\end{equation}
defines a faithful  unitary representation of $U(n)$ on $L^2(\bS^{2n-1},d\nu_n)$.

\begin{lemma}\label{lemmasc}
For every $n=1,2,\ldots$ the unitary representation (\ref{rep}) is strongly continuous.
\end{lemma}

\begin{proof} It is enough proving the continuity at $g=I$.
If $f: \bS^{2n-1}\to \bC$ is continuous, $U(n)\times \bS^{2n-1} \ni (g,u) \mapsto f(g^{-1}u)$ is jointly continuous and thus bounded by $K<+\infty $ since the domain is compact. 
Exploiting Lebesgue dominated convergence theorem as $|f\circ g^{-1}(u) - f(u)|^2 \leq K$ and the constant function $K$ being integrable since
 the measure $\nu_n$ is finite:
$$||D_n(g) f -f||_2^2 = \int_{\bS^{2n-1}} |f\circ g^{-1} - f|^2 d\nu_n \to 0 \quad \mbox{as $g \to I$,}$$
 If $f$ is not continuous, due to Luzin's theorem, there is a sequence of continuous functions $f_n$ converging to $f$ in the norm of $L^2(\bS^{2n-1},d\nu_n)$.
Therefore
$$||f\circ g^{-1}- f||_2\leq ||f\circ g^{-1}- f_n \circ g^{-1}||_2 + ||f_n\circ g^{-1}- f_n||_2 +
||f_n- f||_2\:.$$
If $\epsilon >0$, there exists $k$ with $||f\circ g^{-1}- f_k\circ  g^{-1}||_2= ||f- f_k||_2 <\epsilon/3$
where we have also used the $U(n)$-invariance of $\nu_n$. Since $f_k$ is continuous we can apply the 
previous result getting $||f_k\circ g^{-1}- f_k||_2 < \epsilon/3$ if $g$ is sufficiently close to $I$.   
\end{proof}

We are in a position to define the notion of spherical harmonics we shall use in the rest of the paper.
If, $p,q=0,1,2,\ldots$, $\cP^{p,q}$  denotes the set of polynomials 
$h:\bS^{2n-1}\to \bC$ such that $h(\alpha z_1,...,\alpha z_n)=\alpha^p\overline \alpha^q h(z_1,...,z_n)$ for all $\alpha\in\bC$. The standard Laplacian $\Delta_{2n}$ on $\bR^{2n}$ can be applied to the elements of $\cP^{p,q}$ in terms of decomplexified $\bC^n$. Now, we have the following known result  (see Theorems 12.2.3, 12.2.7 in \cite{rudin} and
 theorem 1.3 in \cite{armgen}):

\begin{theorem}\label{teodec}
If $\cH_{(p,q)}^n:=Ker\Delta_{2n}|_{\cP^{p,q}}$, the following facts hold.\\
{\bf (a)} The orthogonal decomposition is valid,  each $\cH_{(p,q)}^n$ being finite-dimensional and closed:
\begin{equation}
L^2(\bS^{2n-1},d\nu_n)=\bigoplus _{p,q=0}^{+\infty} \cH_{(p,q)}^n.
\end{equation}
{\bf (b)} Every $\cH_{(p,q)}^n$ is invariant and irreducible under the representation (\ref{rep}) of $U(n)$, so that the said representation correspondingly decomposes as
$$D_n(g)=\bigoplus _{p,q=0}^{+\infty} D^{(p,q)}_n(g) \quad \mbox{with $D^{(p,q)}_n(g):= D_n(g)|_{\cH_{(p,q)}^n}$.}$$
{\bf (c)} If $(p,q) \neq (r,s)$ the irreducible representations $D^{p,q}_n$ and $D^{r,s}_n$ are unitarily inequivalent:
 no unitary operator $U: \cH_{(p,q)}^n \to \cH_{(r,s)}^n$ exists such that $UD^{(p,q)}_n(g) = D^{(r,s)}_n(g) U$
for every $g\in U(n)$.
\end{theorem}
\begin{definition}
For $j\equiv (p,q)$, with $p,q=0,1,2,...$, the {\bf generalized complex spherical harmonics} of order $j$
are the elements of  $\cH_{(p,q)}^n$. 
\end{definition}

\noindent A useful technical lemma is the following.

\begin{lemma}\label{lemma1}
For $n\geq 3$, $\cH^n_{(1,1)}$ is made of the restrictions to $\bS^{2n-1}$ of the polynomials
 $h^{(1,1)}(z,\overline z)=\overline z^t A z$, $A$ being any traceless $n\times n$ matrix and $z\in \bC^n$.
\end{lemma}

\begin{proof}
$h^{(1,1)}$ is of first-degree in each variables so
 $h^{(1,1)}(z,\overline z)=\overline z^t A z$ for some $n\times n$ matrix $A$.
$\Delta_{2n} h^{(1,1)}=0$ is equivalent to $trA=0$ as one verifies by direct inspection.
\end{proof}

\noindent For $n\geq 3$, there is a special class of spherical harmonics in $\cH^n_j$ that are parametrized by elements $t\in \bS^{2n-1}$ \cite{armgenfg}.

\begin{definition}\label{zonali}
For $n\geq 3$, the {\bfseries zonal spherical harmonics} are elements of $\cH^n_j$ defined, for every $t \in \bC^n$, as
\begin{equation}\label{zonal}
F_{n,t}^j(u):= R^n_j \left( \overline{u}^t\cdot t \right) \quad \forall u\in \bS^{2n-1},
\end{equation}
where the polynomials $R^n_j(z)$ have the generating function
\begin{equation}\label{gen1}
(1-\xi z- \eta \overline z +\xi \eta )^{1-n}=\sum_{p,q=0}^{+\infty} R^n_{p,q}(z) \xi^p \eta^q 
\end{equation}
with $|z|\leq 1$, $|\eta|<1$, $|\xi|<1$.
\end{definition}
\noindent These zonal spherical harmonics  are a generalization of the eigenfunctions of orbital angular momentum with $L_z$-eigenvalue $m=0$.
From (\ref{gen1}) we get two identities useful later: 
$$p!q!R^n_{p,q}(1)=(-1)^{p+q}(n-1)n(n+1)\cdot \cdot \cdot (n+p-2)(n-1)n(n+1)\cdot \cdot \cdot (n+q-2),$$
\begin{equation}\label{pq}
p!q!R^n_{p,q}(0)=(-1)^p\delta_{pq}p!(n-1)n(n+1)\cdot \cdot \cdot (n+p-2).\quad\quad\quad\quad\quad\quad\quad\quad\quad\quad
\end{equation}

\section{Generalized complex Harmonics and Frame Functions}

To prove our main statement in the next section we need the following preliminary technical result that relies on the technology 
presented in Chapter 7 of \cite{BR}.

\begin{proposition} \label{lemma2} 
If $f\in \cL^2(\bS^{2n-1},d\nu_n)$,  each projection $f_j$ on $\cH_j^n$  
verifies, $\mu$ being the Haar measure on $U(n)$ normalized to $\mu(U(n))=1$: 
\beq
f_j(u) = dim(\cH^n_j) \int_{U(n)} tr(\overline{D^j(g)}) f(g^{-1} u) d\mu(g)\quad \mbox{a.e. in $u$ with respect to $\nu_n$,} \label{newform}
\eeq
where the right-hand side is a continuous function of $u \in \bS^{2n-1}$.\\
If $f\in \cL^2(\bS^{2n-1},d\nu_n)$ is a frame function, then 
$f_j$  (possibly re-defined on a zero-measure set in order to be continuous)
is a frame function as well with $W_{f_j}=0$ when $j\not = (0,0)$.
\end{proposition}

\begin{proof} First of all notice that, if $f\in \cL^2(\bS^{2n-1},d\nu_n)$, the right-hand side of (\ref{newform}) is well defined and continuous as we go to prove. 
$U(n) \ni g \mapsto tr(\overline{D^j(g)})$ is continuous -- and thus bounded since $U(n)$ is compact -- in view of lemma \ref{lemmasc} and  $dim(\cH_j^n)$
is finite for theorem \ref{teodec}. Furthermore, for almost all   $u \in \bS^{2n-1}$ the map $U(n) \ni g \mapsto f(g^{-1}u)$ is $\cL^2(U(n), d\mu)$ --
and thus $\cL^1(U(n), d\mu)$ because the measure is finite -- as follows by Fubini-Tonelli theorem and the invariance of $\nu_n$ under $U(n)$, it being
$$\int_{U(n)} d\mu(g) \int_{\bS^{2n-1}}  |f(g^{-1} u)|^2 d\nu_n(u) = \int_{U(n)} d\mu(g) \int_{\bS^{2n-1}}  |f(u)|^2 d\nu_n(u) = \mu(U(n)) ||f||_2 <+\infty\:.$$ 
Consequently, in view of the fact that $\mu$ is invariant, and $U(n)$ transitively acts on $\bS^{2n-1}$,
 the map $U(n) \ni g \mapsto f(g^{-1}u)$ is $\cL^2(U(n), d\mu)$ (and thus $\cL^1(U(n), d\mu)$) for all $u\in \bS^{2n-1}$.
Continuity in $u$ of the right-hand side of (\ref{newform})  can be proved as follows. Let   $u_0 = [I] \in \bS^{2n-1} \equiv U(n)/U(n-1)$. 
Since $U(n)$ and $U(n-1)$
are Lie groups, for any fixed $u_1 \in \bS^{2n-1}$ there is an open neighbourhood $W_{u_1}$ of $u_1$ and a 
smooth map $W_{u_1} \ni u \mapsto g_u \in U(n)$ such that $[g_u]=u$  (Theorem 3.58 in \cite{Warner}).  As a consequence $g_uu_0 = [g_uI] =[g_u]=u$.  
 Therefore, using the invariance of the Haar measure and for $u = g_uu_0\in W_{u_1}$:
 $$ \int_{U(n)} tr(\overline{D^j(g)}) f(g^{-1} u) d\mu(g) =  \int_{U(n)} tr(\overline{D^j(g_u g)}) f(g^{-1}u_0) d\mu(g)\:.$$
 Since $W_{u_1}\times U(n) \ni (u,g) \mapsto tr(\overline{D^j(g_u g)})$ is continuous due to lemma \ref{lemmasc}, 
 the measure is finite and $g \mapsto f(g^{-1}u_0)$ is integrable, 
 Lebesgue dominated convergence theorem implies that, as said above, $W_{u_1} \ni u \mapsto \int_{U(n)} tr(\overline{D^j(g_u g)}) f(g^{-1}u_0) d\mu(g)$ 
 is continuous in $u_1$
 and thus everywhere on $\bS^{2n-1}$ since $u_1$ is arbitrary.\\
Let us pass to prove (\ref{newform}) for $f$ containing a finite number of components. So $F$ is finite, $f\in \cL^2(\bS^{2n-1},d\nu_n)$ and:
$$f(u)= \sum_{j\in F} f_j(u) = \sum_{j\in F} \sum_{m=1}^{dim(\cH_j^n)} f^{j}_m Z_{m}^j(u)\:,\qquad   f^{j}_m \in \bC$$
where $\{Z_{m}^j\}_{m=1,...,dim(\cH_j^n)}$ 
is  an orthonormal basis of $\cH_j^n$, with $Z_n^{(0,0)}=1$, made of continuous functions
(it exists in view of the fact that $\cP^{p,q}$ is a space of polynomials and exploiting  Gramm-Schmitd's procedure).
Then
$$\overline{D^{j_0}_{m_0m'_0}(g)}f(g^{-1}u)= \sum_{j\in F} \sum_{m,m'} \overline{D^{j_0}_{m_0m'_0}(g)}  D^{j}_{mm'}(g)  f^{j}_{m'} Z_{m}^j(u)\:.$$
In view of (c) of theorem \ref{teodec} and Peter-Weyl theorem, taking the integral over $g$ with respect to the Haar measure on $U(n)$ one has:
$$\int \overline{D^{j_0}_{m_0m'_0}(g)}f(g^{-1}u) d \mu(g) =  dim(\cH_{j_0}^n) \:f^{j_0}_{m'_0} Z_{m_0}^{j_0}(u)\:,$$
that implies (\ref{newform}) when taking the trace, that is summing over $m_0=m'_0$. To finish with the first part, let us generalize the obtained formula to the 
case of $F$ infinite. In the following $P_j : L^2(\bS^{2n-1}, \nu_n) \to  L^2(\bS^{2n-1}, \nu_n)$ is the orthogonal projector onto $\cH^n_j$. 
The convergence in the norm $||\:\:||_2$ implies that in the norm 
$||\:\:||_1$, since $\nu_n(\bS^{2n-1})< +\infty$. So if $h_m \to f$ in the norm $||\:\:||_2$, as 
$P_j$ is bounded:
$${\lim_{m\to +\infty}}^{(1)} P_jh_m = {\lim_{m\to +\infty}}^{(2)} P_jh_m
= P_j\left({\lim_{m\to +\infty}}^{(2)}h_m\right) = P_jf\:.$$
We specialize to the case $h_m= \sum_{(p,q)=(0,0)}^{p+q=m} f_{(p,q)}$ so that $h_m \to f$ as $m\to +\infty$ in the norm $||\:\:||_2$. As every $h_m$ 
has a finite number of harmonic components the identity above yields:
$$ dim(\cH^n_j)  {\lim_{m\to +\infty}}^{(1)}\int_{U(n)} tr(\overline{D^j(g)}) h_m(g^{-1} u) d\mu(g) =
P_j f =: f_j\:.$$
Now notice that, as $U(n)\ni g \mapsto tr(\overline{D^j(g)})$ is bounded on $U(n)$ by some $K<+\infty$:
$$\left| \left| \int_{U(n)} tr(\overline{D^j(g)}) h_m(g^{-1} u) d\mu(g) - \int_{U(n)} tr(\overline{D^j(g)}) f(g^{-1} u) d\mu(g)\right|\right|_1$$
$$\leq K \int_{\bS^{2n-1}} d\nu(u) \int_{U(n)} d\mu(g)\left| h_m(g^{-1} u)-  f(g^{-1} u)\right| $$
$$=K \int_{U(n)} d\mu(g)  \int_{\bS^{2n-1}} d\nu(u)\left| h_m(g^{-1} u)-  f(g^{-1} u)\right|$$
$$=K \int_{U(n)} d\mu(g)  \int_{\bS^{2n-1}} d\nu(u)\left| h_m( u)-  f( u)\right| = K \mu(U(n))||h_m-f||_1 \to 0\:.$$
We have found that, as desired, (\ref{newform}) holds for  $f$, because
$$\left| \left|f_j  - dim(\cH^n_j) \int_{U(n)} tr(\overline{D^j(g)}) f(g^{-1} u) d\mu(g) \right| \right|_1=0\:.$$
To prove the last statement, assume $j\neq (0,0)$ otherwise the thesis is trivial (since $f_{(0,0)}$ is a constant function). 
We notice that, when $f_j$ is taken to be continuous (and it can be done in a unique way in view of lemma \ref{lemma0}, referring the the Hilbert basis 
of continuous functions $Z^j_m$ as before), the identity (\ref{newform}) must be true for every $u\in \bS^{2n}$. Therfore, if $e_1,e_2,\ldots, e_n$ is any Hilbert basis of $\bC^n$
$$\frac{1}{dim(\cH^n_j)}\sum_k f_j(e_k) =  \int_{U(n)} tr(\overline{D^j(g)}) \sum_k f(g^{-1} e_k) d\mu(g) =  \int_{U(n)} tr(\overline{D^j(g)}) W_f d\mu(g)=0$$
 because $W_f$ is a constant and thus it is proportional to $1=D^{(0,0)}$ which, in turn, is orthogonal to $D^j_{mm'}$ for $j\neq (0,0)$ in view of Peter-Weyl 
 theorem and (c) of theorem \ref{teodec}.
\end{proof}

\section{The main result} If $\cH$ is a finite-dimensional complex Hilbert space $\cH$, with $\dim(\cH)=n\geq 3$,
there is only a regular Borel measure, $\nu_\cH$, on $\bS(\cH)$ which is left-invariant under the natural action of every unitary operator
 $U : \cH \to \cH$ and
$\nu_\cH(\bS(\cH))=1$.
It is the $U(n)$-invariant measure $\nu_n$ induced by any identification of $\cH$ with a corresponding $\bC^n$ obtained by fixing 
a orthonormal basis in $\cH$. The uniqueness of $\nu_\cH$ is due to the fact that different orthonormal bases are connected by means of transformations in $U(n)$.

\begin{theorem}\label{ZAZ}
If  $f: \bS(\cH)\to \bC$ is  a  frame function on a finite-dimensional complex Hilbert space $\cH$, with $\dim(\cH) \geq 3$ and
$f \in \cL^2(\bS(\cH) ,d\nu_\cH)$, then
there is a unique linear operator $A : \cH \to \cH$ such that:
\begin{equation}
f(z)=\langle z|A z \rangle \quad \forall z\in \bS(\cH),
\end{equation}
where $\langle \,\, |\,\,\rangle$ is the inner product in $\cH$. $A$ turns out to be Hermitean if $f$ is real.
\end{theorem}    

\noindent{\bf Remark}. Since $\nu_\cH$ is finite,  $f \in \cL^2(\bS(\cH) ,d\nu_\cH)$ holds in particular when
$f \in \cL^p(\bS(\cH),d\nu_\cH)$ for some $p \in [2,+\infty]$, as a classic result based on Jensen's inequality.\\

\begin{proof} 
We start from the uniqueness issue. Let $B$ be another  operator  satisfying the thesis,
so that  $\langle z|(A-B)z\rangle =0$ $\forall z \in \cH$.
Choosing $z=x+y$ and then $z=x+iy$ one finds $\langle x|(A-B)y\rangle =0$ for every $x,y \in \cH$, that is
$A=B$. We pass to the existence of $A$ identifying $\cH$ to $\bC^n$  by means of an orthonormal basis $\{e_k\}_{k=1,...,n}\subset \cH$.
As $f\in \cL^{2}(\bS^{2n-1},d\nu_n)$, $f$ can be decomposed as:
 $f=\sum_j f_j$ with $f_j\in \cH_j^n$. Lemma \ref{lemma2} implies  that, if $g\in U(n)$:
\begin{equation}\label{D}
\sum_{k=1}^n \left ( D^j(g) f_j \right )(e_k)=\sum_{k=1}^n  f_j(g^{-1}e_k)=0 \quad \mbox{if} \quad j\not = (0,0)\:
\end{equation}  
 Assuming $f_j\not =0$, since the representation $D^j$ is irreducible, the subspace of $\cH_j^n$ 
 spanned by all the vectors $D^j(g)f_j\in \cH_j^n$ is  dense in  $\cH_j^n$ when $g$ ranges in $U(n)$. As $\mathcal H_j^n$ 
 is finite-dimensional, the dense subspace is  $\cH_j^n$ itself. So it must be $\sum_{k=1}^n Z(e_k)=0$ for every $Z\in \cH_j^n$.
  In particular it holds for the  zonal spherical harmonic $F_{n,e_1}^j$ 
 individuated by $e_1$:
$\sum_{k=1}^nF_{n,e_1}^j(e_k)=0$.
By definition of zonal spherical harmonics the above expression can be written in these terms: $R^n_{p,q}(1)+(n-1)R^n_{p,q}(0)=0$,
and using relations (\ref{pq}):
$$(-1)^{p+q}(n-1)n(n+1)\cdot \cdot \cdot (n+p-2)(n-1)n(n+1)\cdot \cdot \cdot (n+q-2)=$$
\begin{equation}\label{pqp}
=(-1)^p\delta_{pq} p! (n-1)^2n(n+1)\cdot \cdot \cdot (n+p-2).\quad \quad \quad \quad \quad \quad \quad \quad \quad
\end{equation}
(\ref{pqp}) implies $p=q$. Indeed, if $p\not =q$  the right hand side vanishes, while the left 
does not.
Now, for $n\geq 3$ and $j\not =(0,0)$ we can write:
\begin{equation}\label{np}
(n-1)^2n^2(n+1)^2\cdot \cdot \cdot (n+p-2)^2=(-1)^p p!(n-1)^2 n(n+1)\cdot \cdot \cdot (n+p-2).
\end{equation}
The identity (\ref{np}) is verified if and only if $p=1$. In view of lemma \ref{lemma1}, we know that the functions
 $f_{(1,1)}\in \cH^n_{(1,1)}$ have form $f(x)= \langle x|A_0x\rangle$ with $trA_0=0$.
We conclude that our frame function $f$ can only have the form:
$$f(x)=c+f_{(1,1)}(x)= \langle x|cIx\rangle +\langle x|A_0x\rangle =\langle x|A x\rangle  \quad x	\in \bS^{2n-1}\:.$$
If $f$ is real valued  $\langle x|Ax\rangle = \overline{\langle x|Ax\rangle} = \langle x|A^*x\rangle$ and thus
$\langle x|(A-A^*)x\rangle =0$. Exploiting the same argument as that used in the proof of the uniqueness, we conclude that
$A=A^*$.
\end{proof}

\section*{Acknowledgements}
The authors are grateful to Alessandro Perotti for useful comments and suggestions.

\end{document}